\documentclass[journal,twocolumn,twoside,times]{IEEEtran}

\usepackage{cite,graphicx,amsmath,amsfonts,amssymb,psfrag}
\usepackage{color,mathrsfs}
\usepackage[nolist]{acronym}

\newcommand{\scrE}{\mathscr{E}}
\newcommand{\scrF}{\mathscr{F}}
\newcommand{\bfb}{\mbox{\boldmath $b$}}
\newcommand{\bfd}{\mbox{\boldmath $d$}}
\newcommand{\bfzero}{\mbox{\boldmath $0$}}
\newcommand{\bfone}{\mbox{\boldmath $1$}}
\newcommand{\bfonee}{\mbox{\boldmath $1$}_e}

\newcommand{\bff}{\mbox{\boldmath $f$}}
\newcommand{\bfS}{\mathbf{S}}
\newcommand{\bfG}{\mathbf{G}}

\newtheorem{theorem}{\indent Theorem}[section]

\newtheorem{corollary}{\indent Corollary}[section]
\newtheorem{lemma}{\indent Lemma}[section]
\newtheorem{example}{\indent Example}[section]

\newtheorem{remark}{\indent Remark}[section]

\def\rank{\mathrm{rank}}
\renewcommand{\QED}{\QEDopen}

\title{Stability of Iterative Decoding of Multi-Edge Type Doubly-Generalized LDPC Codes Over the
BEC}
\author{Enrico Paolini, Mark F. Flanagan, Marco Chiani and Marc P. C. Fossorier
\thanks{E. Paolini and M. Chiani are with DEIS/WiLAB, University of Bologna, via Venezia 52, 47521
Cesena (FC), Italy (e-mail: e.paolini@unibo.it, marco.chiani@unibo.it). 
M.F. Flanagan is with the School of Electrical, Electronic and Mechanical Engineering,
University College Dublin (UCD), Belfield, Dublin 4, Ireland (e-mail: mark.flanagan@ieee.org). 
M.P.C. Fossorier is with ETIS ENSEA, UCP, CNRS UMR-8051, 6 avenue du Ponceau, 95014 Cergy
Pontoise, France (e-mail:  mfossorier@ieee.org).}
\thanks{Supported in part by the EC under Project FP7 OPTIMIX (ICT-214625). 
}
}

\begin{document}
\maketitle

\begin{abstract}
Using the EXIT chart approach, a necessary and sufficient condition is developed for the local
stability of iterative decoding of multi-edge type (MET) doubly-generalized low-density parity-check
(D-GLDPC) code ensembles. In such code ensembles, the use of arbitrary linear block codes as
component codes is combined with the further design of local Tanner graph connectivity through the
use of multiple edge types. The stability condition for these code ensembles is shown to be
succinctly described in terms of the value of the spectral radius of an appropriately defined
polynomial matrix.
\end{abstract}

\section{Introduction}\label{sec:introduction}
Multi-edge type (MET) low-density parity-check (LDPC) codes were originally proposed in
\cite{richardson2004:MET} as a framework to capture both degree-$1$ variable nodes (VNs) and
punctured bits in the analysis of LDPC code ensembles, and to achieve a finer control of the
connectivity between VNs and check nodes (CNs) in the ensemble definition. The new framework
allowed the design of powerful finite-length LDPC codes over the additive white Gaussian noise
(AWGN) channel, with a very good compromise between waterfall performance and error-floor. Since
then, several aspects of MET LDPC codes have been investigated such as, for instance, their average
weight distribution \cite{kasai10:weight_MET}. Traditional unstructured irregular LDPC code
ensembles parametrized through their degree distribution pair \cite{luby01:improved} may be seen as
MET ensembles where all edges in the Tanner graph are of the same type. On the other hand, LDPC
codes based on protographs \cite{thorpe2003:protograph} may be seen as MET LDPC codes such that no
two edges connected to the same VN or to the same CN are of the same type.

Another way to extend the original framework of unstructured LDPC code ensembles consists of
replacing the VNs and the CNs with linear block codes other than repetition codes and single
parity-check (SPC) codes, respectively. The resulting LDPC-like codes are called doubly-generalized
LDPC (D-GLDPC) codes \cite{wang06:D-GLDPC}, and extend the original idea of generalized LDPC (GLDPC)
codes \cite{tanner81:recursive}, where only the CNs were replaced with generic linear block codes.
Several theoretical aspects of unstructured D-GLDPC codes have been recently clarified, such as
their stability condition over the binary erasure channel (BEC) \cite{paolini09:stability}, and the
analysis of the exponent of their weight distribution \cite{flanagan11:growth_rate}.

The two different extensions can be considered together, leading to the concept of MET D-GLDPC code
ensemble. This represents a very general framework for design and analysis of LDPC-like codes,
enabling to handle different variable and check component codes along with puncturing; VNs and CNs
with local minimum distance $1$, including degree-$1$ VNs (state variables), can be also
considered. The asymptotic weight enumerators for MET D-GLDPC codes were investigated in
\cite{fossorier09:MET}, while EXIT analysis to calculate the threshold of MET D-GLDPC codes over
the BEC was developed in \cite{paolini10:multi}.

In this paper, we analyze the convergence properties of the belief-propagation (BP) decoder for MET
D-GLDPC code ensembles over the BEC by developing its \emph{stability condition}, i.e., the
condition under which the erasure-free state attracts the decoder, in the asymptotic setting where
the codeword length tends to infinity. If and only if the condition is satisfied, BP decoding can
in principle succeed, provided the BEC erasure probability is below the threshold which can be
calculated using the technique in \cite{paolini10:multi}. The stability condition is obtained in
the case where the are no punctured bits and where the local minimum distance of each VN and CN is
at least $2$. It can be shown that the obtained condition coincides with that developed in
\cite{richardson2004:MET},\cite[Ch.~7]{urbanke:modern} in the special case of MET LDPC~codes.

\section{Preliminary Definitions}\label{sec:preliminary_def}
\subsection{Concept of D-GLDPC Codes}\label{subsec:DGLDPC_def}
A D-GLDPC code consists of a set of CNs and a set of VNs. Each CN
corresponds to some arbitrary linear `local' code. On the other hand, each VN corresponds to some
arbitrary linear `local' code, together with its \emph{encoder} (i.e., generator matrix).
Graphically, each CN and each VN may be viewed as having a set of \emph{sockets} corresponding to
the bits in the local codeword. The sockets of the CNs are connected by edges to the sockets of the
VNs in a one-to-one fashion; the resulting graph is called the \emph{Tanner graph} of the D-GLDPC
code. A \emph{codeword} of the D-GLDPC code is an assignment of an information word to each VN such
that the local encoding of this word at each VN assigns an encoded bit to each edge of the Tanner
graph in such a way that the resulting configuration forms a valid local codeword from the
perspective of \emph{every} CN. It is easily seen that if the local code at each CN is a
\emph{single parity-check} code and if the local code at each VN is a \emph{repetition} code,
the resulting D-GLDPC code reduces to an ordinary LDPC code. 

\subsection{MET D-GLDPC Code Ensemble Definition}\label{subsec:MET_ensemble_def}
In MET D-GLDPC codes, we distinguish between $n_e$ different edge types. Each edge type is
identified by an index in the set $\scrE=\{1,2,\dots,n_e\}$. Furthermore, we distinguish between
different VN types and different CN types. Each VN type is identified by a triplet
$\gamma=(v_{\gamma},\bfb_{\gamma},\bfd_{\gamma})$, where: 
\begin{itemize}
\item $v_{\gamma}$ identifies a $(q_{\gamma},k_{\gamma})$ variable component code, where
$q_{\gamma}$ is the code length and $k_{\gamma}$ the code dimension, and a specific
encoder (i.e., generator matrix $\bfG_{\gamma}$) of it;
\item $\bfb_{\gamma}$ is a binary vector of length $k_{\gamma}$ which specifies the local puncturing
pattern for the VN. Specifically, for $i\in\{1,2,\dots,k_{\gamma}\}$, if $b_{\gamma,i}=0$ then the
corresponding encoded
bit of the D-GLDPC code is punctured, and it is not punctured otherwise;
\item $\bfd_{\gamma}$ is a vector of length $q_{\gamma}$ whose $i$-th element $d_{\gamma,i}\in\scrE$
specifies the edge type of the $i$-th VN socket.
\end{itemize}
Each CN type is identified by a pair
$\delta=(c_{\delta},\bfd_{\delta})$, where:
\begin{itemize}
\item $c_{\delta}$ identifies an $(s_{\delta},h_{\delta})$ check component code (regardless of its
representation), where $s_{\delta}$ is the code length and $h_{\delta}$ the code dimension; 
\item $\bfd_{\delta}$ is a vector of length $s_{\delta}$ whose $i$-th element
$d_{\delta,i}\in\scrE$ specifies the edge type of the $i$-th CN socket.
\end{itemize}

The set of all VN types $\gamma$ is denoted by $\scrF_V$, and set of all CN types $\delta$ by
$\scrF_C$. Moreover, the fraction of edges of type $l\in\scrE$ connected to VNs of type $\gamma\in
\scrF_V$ is denoted by $\lambda_{\gamma,l}$, while the fraction of edges of type $l\in\scrE$
connected to CNs of type $\delta\in \scrF_C$ by $\rho_{\delta,l}$. We have $\lambda_{\gamma,l}>0$
if and only if the generic VN of type $\gamma$ has at least one socket of type $l$, and
$\lambda_{\gamma,l}=0$ otherwise. An analogous statement can be made regarding $\rho_{\delta,l}$.
Also, $q_{\gamma,l}$ and $s_{\delta,l}$ denote the number of sockets of type $l$ for a VN of type
$\gamma$ and for a CN of type $\delta$, respectively. The constraints $\sum_{l\in\scrE} q_{\gamma,l}
= q_{\gamma}$ $\forall \gamma \in \scrF_V$ and $\sum_{l\in\scrE} s_{\delta,l} = s_{\delta}$ $\forall
\delta \in \scrF_C$ hold. 

An example of D-GLDPC code ensemble is depicted in
Fig.~\ref{fig:ensemble}. As opposed to single-edge type codes, where a unique edge interleaver for
all edges is present, for MET codes a dedicated edge interleaver is present for all edges of the
same type. For a given codeword length, each code in the ensemble corresponds to a specific
realization of the $n_e$ edge interleavers, where all realizations of each edge interleaver are
equiprobable.

In the following, we make the assumption that there are no VNs and CNs with local minimum distance
$1$. We also make the assumption that no encoded bit of the D-GLDPC code is punctured, i.e., for
each VN type $\gamma=(v_{\gamma},\bfb_{\gamma},\bfd_{\gamma})\in\mathscr{F}_V$ the vector
$\bfb_{\gamma}$ has no
`$0$' entries. Finally, we assume that no variable or check component code has idle bits. 

\begin{figure}[t]
\begin{center}
\includegraphics[width=0.9\columnwidth]{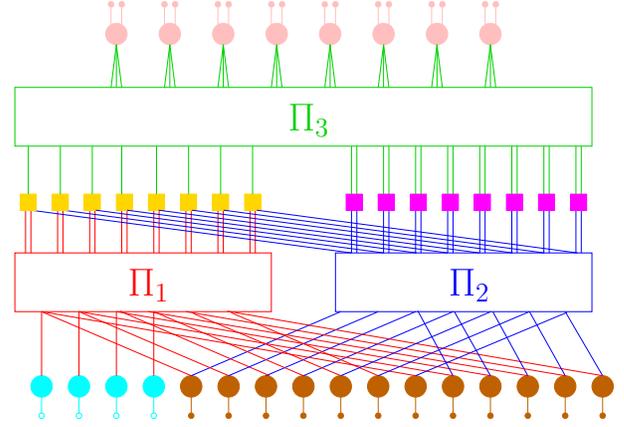}
\end{center}
\caption{Example of a MET D-GLDPC code ensemble. The set of edge types is
$\scrE=\{1,2,3\}$, where edges of type $1$, $2$ and $3$ are depicted in red, blue and green,
respectively. A separate edge-interleaver is present for each edge type. The set of VN types is
$\scrF_V=\{\gamma_1,\gamma_2,\gamma_3\}$, where the VNs of type $\gamma_1$, $\gamma_2$ and
$\gamma_3$ are represented as cyan, brown and pink circles, respectively. The VNs of type $\gamma_1$
are $(1,1)$ codes, those of type $\gamma_2$ are $(2,1)$ codes, and those of type $\gamma_3$ are
$(3,2)$ codes. Each VN of type $\gamma_1$ and $\gamma_2$ is associated with one encoded bit of the
D-GLDPC code, each VN of type $\gamma_3$ with two encoded bits of the D-GLDPC code. The encoded bits
corresponding to type-$\gamma_1$ VNs are punctured. The set of CN types is
$\scrF_C=\{\delta_1,\delta_2\}$, where the CNs of type $\delta_1$ and $\delta_2$ are represented as
yellow and magenta squares, respectively. If each CN of type $\delta_1$ introduces two parity-check
equations and each CN of type $\delta_2$ one parity-check equation, then the overall code dimension
is $K=8$. Since four bits are punctured, the overall codeword length is $N=28$, and the code rate is
$R=2/7$.}
\label{fig:ensemble}
\end{figure}

\subsection{Further Definitions}\label{subsec:further_def}
Throughout the paper, vectors are intended as column vectors. We define $\bfzero$ and $\bfone$ as
the vectors of length $n_e$ whose elements are all equal to $0$ and all equal to $1$, respectively.
Morever, we define $\bfone_e$ as the vector of length $n_e$ whose elements are all equal to $0$
except the element in position $e$ which is equal to $1$. The subset of VN types $\gamma$ with local
minimum distance $2$ is denoted by $\scrF_{V2}\subseteq \scrF_V$, and the subset of CN types
$\delta$ with local minimum distance $2$ by $\scrF_{C2}\subseteq \scrF_C$. 

For $\delta\in \scrF_{C2}$ and $l,m\in\scrE$, we denote by $\xi_{2}^{(\delta)}(l,m)$ the number of
\emph{ordered pairs} of sockets of a CN of type $\delta$, such that the first socket is of
type $l$ and the second of type $m$, and such that the assignment of a `$1$' to these sockets and a
`$0$' to all other CN sockets results in a (weight-$2$) local codeword. Note that
$\xi_{2}^{(\delta)}(l,m)=\xi_{2}^{(\delta)}(m,l)$. For $l,m\in\scrE$, we define the nonnegative real
parameter $C^{l,m}$ as 
\begin{align}\label{eq:Clm}
C^{l,m} := \sum_{\delta\in\scrF_{C2}} \left(\frac{\rho_{\delta,l}}{s_{\delta,l}}\right)
\xi_{2}^{(\delta)}(l,m)\,.
\end{align}
If CNs of type $\delta\in\scrF_{C2}$ have no sockets of type $l$ (in which case $\rho_{\delta,l}=0$
and $s_{\delta,l}=0$), then we set $C^{l,m}=0$ by definition. Denoting by $A_2^{(\delta)}(l,m)$ the
number of weight-$2$ local codewords of a CN of type $\delta\in\scrF_{C2}$ such that one of the two
`$1$' local
encoded bits corresponds to a socket of edge type $l$ and the other to a socket of edge type $m$, we
have $\xi_{2}^{(\delta)}(l,m)=A_{2}^{(\delta)}(l,m)$ for $l\neq m$ and $\xi_{2}^{(\delta)}(l,l)=2
A_{2}^{(\delta)}(l,l)$.

For $\gamma\in\scrF_{V2}$ and $l,m\in\scrE$, we denote by $\chi_{2,u}^{(\gamma)}(l,m)$ the number of
\emph{ordered pairs} of sockets of a VN of type $\gamma$ such that the first socket is of
type $l$ and the second of type $m$, and such that the assignment of a `$1$' to these sockets and a
`$0$' to all other VN sockets results in a (weight-$2$) local codeword generated by a local input
word of weight $u$. Similarly to the CN case, we have
$\chi_{2,u}^{(\gamma)}(l,m)=\chi_{2,u}^{(\gamma)}(m,l)$. The nonnegative polynomial $P^{l,m}(x)$
(with real coefficients) is defined as
\begin{align}\label{eq:Plm}
P^{l,m}(x) := \sum_{\gamma\in\scrF_{V2}} \left(\frac{\lambda_{\gamma,l}}{q_{\gamma,l}}\right)
\sum_{u=1}^{k_{\gamma}}\chi_{2,u}^{(\gamma)}(l,m)\, x^u\,.
\end{align}
If VNs of type $\gamma\in\scrF_{V2}$ have no sockets of type $l$ (in which case
$\lambda_{\gamma,l}=0$ and $q_{\gamma,l}=0$), then we set $P^{l,m}(x)=0$ by definition.
Moreover, denoting by $B_{2,u}^{(\gamma)}(l,m)$ the number of weight-$2$ local codewords of a VN of
type $\gamma\in\scrF_{V2}$ generated by local weight-$u$ input words, and such that one of the two
`$1$' local
encoded bits corresponds to a socket of edge type $l$ and the other to a socket of edge type $m$, we
have $\chi_{2,u}^{(\gamma)}(l,m)=B_{2,u}^{(\gamma)}(l,m)$ for $l\neq m$ and
$\chi_{2,u}^{(\gamma)}(l,l)=2 B_{2,u}^{(\gamma)}(l,l)$.

\medskip
\begin{remark} For $l,m\in\scrE$ and $l \neq m$, in general we have $C^{l,m} \neq C^{m,l}$ and
$P^{l,m}(x) \neq P^{m,l}(x)$.
\end{remark}

\medskip
\begin{remark} For single-edge type codes, $\scrE$ is a singleton $\scrE=\{l\}$, $P^{l,l}(x)$
reduces to the polynomial $P(x)$ and $C^{l,l}$ to the parameter $C$ characterizing ordinary D-GLDPC
codes \cite{paolini09:stability,flanagan11:growth_rate}. For codes constructed from protographs, no
two sockets of a VN or CN are of the same type. Hence, in this case we have $P^{l,l}(x)=0$ and
$C^{l,l}=0$ for all $l\in\scrE$.
\end{remark}

\subsection{Multi-Type Information Functions}\label{subsec:multi_type_IF}
Although in this paper we make some assumptions on the VNs and CNs (see the last paragraph of
Sec.~\ref{subsec:MET_ensemble_def}), the definitions provided in this subsection are more general
and do not rely on such assumptions.

Consider a CN of type $\delta=(c_{\delta},\bfd_{\delta})\in\scrF_C$, and let $\bfG_{\delta}$ be
\emph{any} generator matrix for the associated component code. From $\bfG_{\delta}$, form $n_e$
matrices $\bfG_{\delta,l}$, where $\mathbf{G}_{\delta,l}$ is the $(h_{\delta} \times s_{\delta,l})$
matrix composed of the columns of $\mathbf{G}_{\delta}$ associated with the bit positions of type
$l\in\scrE$ (irrespective of the order of these columns). Then, for any integer $n_e$-tuple
$\mathbf{g}=(g_1,g_2,\dots,g_{n_e})$ satisfying $0 \leq g_l \leq s_{\delta,l}$ for all $l\in\scrE$,
the $\mathbf{g}$-th multi-type information function of the CN is defined as 
\begin{align}\label{eq:mt_if_definition}
\tilde{e}^{(\delta)}_{\mathbf{g}} := \sum_{\bfS^{(\delta)}_{\mathbf{g}}}
\rank\left(\bfS^{(\delta)}_{\mathbf{g}}\right)
\end{align}
where $\bfS^{(\delta)}_{\mathbf{g}}$ is a matrix formed by selecting $g_l$ columns in
$\bfG_{\delta,l}$ (irrespective of the order of these columns) and where
$\sum_{\bfS^{(\delta)}_{\mathbf{g}}}$ denotes the summation over all $\prod_{l=1}^{n_e}
{s_{\delta,l} \choose g_l}$ matrices $\bfS^{(\delta)}_{\mathbf{g}}$.\footnote{If for some
$l\in\scrE$ the CN has no sockets of type $l$, then $g_l$ is conventionally set to $0$. This
convention shall be adopted also for the multi-type split information functions defined for the
VNs.}

For a VN of type $\gamma=(v_{\gamma},\bfb_{\gamma},\bfd_{\gamma})$, let $\bfG_{\gamma}$ be the
\emph{specific} generator matrix identified by $v_{\gamma}$. Moreover, let $b_{\gamma}$ be the
Hamming weight of $\bfb_{\gamma}$. From $\bfG_{\gamma}$, form $n_e$
matrices $\bfG_{\gamma,l}$, where $\bfG_{\gamma,l}$ is the $(k_{\gamma} \times q_{\gamma,l})$
matrix composed of the columns of $\bfG_{\gamma}$ associated with the bit positions of type
$l\in\scrE$ (irrespective of the order of these columns). Then, for any integer $n_e$-tuple
$\mathbf{g}=(g_1,g_2,\dots,g_{n_e})$ satisfying $0 \leq g_l \leq q_{\gamma,l}$ for all $l\in\scrE$,
and for any integer $0 \leq u \leq k_{\gamma}$, the $(\mathbf{g};u)$-th multi-type split information
function of the VN is~defined~as 
\begin{align}\label{eq:mt_sif_definition}
\tilde{e}^{(\gamma)}_{\mathbf{g};u} := \sum_{\bfS^{(\gamma)}_{\mathbf{g}; u}}
\rank\left(\bfS^{(\gamma)}_{\mathbf{g};u}\right)
\end{align}
where $\bfS^{(\gamma)}_{\mathbf{g};u}$ is a matrix formed by selecting $g_l$ columns in
$\bfG_{\gamma,l}$ (irrespective of the order of these columns) and $u$ columns among the
$b_{\gamma}$ columns of $\mathbf{I}_{k_{\gamma}}$ (order-$k_{\gamma}$ identity matrix) corresponding
to the support of $\mathbf{b}_{\gamma}^T$. In \eqref{eq:mt_sif_definition},
$\sum_{\bfS^{(\gamma)}_{\mathbf{g};u}}$ denotes the summation over all ${b_{\gamma} \choose u}
\prod_{l=1}^{n_e} {q_{\gamma,l} \choose g_l}$ matrices $\bfS^{(\gamma)}_{\mathbf{g};u}$.

While the $\mathbf{g}$-th multi-type information function of a type-$\delta$ CN is
independent of the specific choice of $\bfG_{\delta}$, the $(\mathbf{g};u)$-th multi-type
split information function of a type-$\gamma$ VN depends on the local mapping between information
and encoded bits defined by $\bfG_{\gamma}$. It also depends on the local puncturing pattern
defined by $\bfb_{\gamma}$.

\section{Exit Analysis and BP Decoder Stability}

\medskip
EXIT analysis of a MET D-GLDPC code ensemble with $n_e$ edge types consists of modeling the average
behavior of the iterative decoder, in the asymptotic case where the codeword length tends to
infinity, through an $n_e$-dimensional discrete dynamical system tracking $n_e$ average
extrinsic information values, one for each edge type. For $e\in\scrE$, the $e$-th
value we track is the average extrinsic information over the edges of type $e$, outgoing
from the VN set towards the~CN~set. 

Let $\ell\geq 1$ denote the decoding iteration index. Let $I^{\ell}_{EV,e}$ and
$I^{\ell}_{EC,e}$ be the average extrinsic information over the edges of type $e$ outgoing from the
VN set and from the CN set, at the $\ell$-th decoding iteration, respectively. Moreover, let
$I^{\ell}_{AV,e}$ and $I^{\ell}_{AC,e}$ be the average \emph{a priori} information over the
edges of type $e$ incoming towards the VN set and towards the CN set, at the $\ell$-th decoding
iteration, respectively. EXIT analysis equations of a MET D-GLDPC code ensemble over a BEC with
erasure probability $\epsilon$ may be expressed as
\begin{align}\label{eq:exit_VN}
\left\{\begin{array}{ccc}
I^{\ell}_{EV,1} & =
& \mathsf{I}_{\mathsf{EV}1}(I^{\ell-1}_{AV,1},I^{\ell-1}_{AV,2},\dots,I^{\ell-1}_{AV,{n_e}},
\epsilon) \\
 & \cdots & \\
I^{\ell}_{EV,{n_e}} &
= & \mathsf{I}_{\mathsf{EV}n_e}(I^{\ell-1}_{AV,1},I^{\ell-1}_{AV,2},\dots,I^{\ell-1}_{AV,{
n_e}},\epsilon)
\end{array}\right.
\end{align}
and
\begin{align}\label{eq:exit_CN}
\left\{\begin{array}{ccc}
I^{\ell}_{EC,1} & =
& \mathsf{I}_{\mathsf{EC}1}(I^{\ell}_{AC,1},I^{\ell}_{AC,2},\dots,I^{\ell}_{AC,{n_e}}) \\
 & \cdots & \\
I^{\ell}_{EC,{n_e}} &
= &
\mathsf{I}_{\mathsf{EC}n_e}(I^{\ell}_{AC,1},I^{\ell}_{AC,2},\dots,I^{\ell}_{AC,{n_e}})
\end{array}\right. \, .
\end{align}
The $2 n_e$ equations \eqref{eq:exit_VN} and \eqref{eq:exit_CN}, together with
$I^{\ell}_{AC,i}=I^{\ell}_{EV,i}$ $\forall i\in\scrE$, $I^{\ell}_{AV,i}=I^{\ell}_{EC,i}$ $\forall
i\in\scrE$, and $I^{0}_{AV,i}=0$ $\forall i\in\scrE$, define a recursion that can be expressed in
the compact form
\begin{align}\label{eq:exit_compact}
\mbox{\boldmath $I$}^{\ell+1}_{EV} = \mbox{\boldmath $f$}(\mbox{\boldmath
$I$}^{\ell}_{EV},\epsilon)
\end{align}
for $\ell\geq 0$ and where $\mbox{\boldmath $I$}_{EV}=[I_{EV,1},I_{EV,2},\dots,I_{EV,{n_e}}]^T$ is
a column vector whose elements are the $n_e$ values to be tracked. The $n_e$-dimensional discrete
dynamical system \eqref{eq:exit_compact} models the asymptotic (in terms of codeword length)
evolution of the BP decoder over a BEC with erasure probability $\epsilon$. The
function $\mbox{\boldmath $f$}(\cdot)$ can be evaluated exploiting results developed
in \cite{paolini10:multi}. In more detail, neglecting the iteration index $\ell$, for $e\in\scrE$ we
have
$\mathsf{I}_{\mathsf{EV},e}(I_{AV,1},I_{AV,2},\dots,I_{AV,n_e},\epsilon)=\sum_{\gamma\in\scrF_V}
\lambda_{\gamma,e} \,
\mathsf{I}_{\mathsf{EV},e}^{(\gamma)}(I_{AV,1},I_{AV,2},\dots,I_{AV,n_e},\epsilon)$
where
\begin{align}\label{eq:GC101}
\mathsf{I}_{\mathsf{EV},e}^{(\gamma)} & (I_{AV,1},I_{AV,2},\dots,I_{AVn_e},\epsilon) = 1 -
\frac{1}{q_{\gamma,e}} \sum_{z=0}^{b_{\gamma}} \epsilon^z
(1-\epsilon)^{b_{\gamma}-z} \notag \\
\, & \times \sum_{t_1=0}^{q_{\gamma,1}} (1-I_{AV,1})^{t_1}
(I_{AV,1})^{q_{\gamma,1}-t_1} \cdots \notag \\
\, & \times \sum_{t_e=0}^{q_{\gamma,e}-1} (1-I_{AV,e})^{t_e}
(I_{AV,e})^{q_{\gamma,e}-1-t_e} \cdots \notag \\
\, & \times \sum_{t_{n_e}=0}^{q_{\gamma,n_e}} (1-I_{AV,n_e})^{t_{n_e}}
(I_{AV,n_e})^{q_{\gamma,n_e}-t_{n_e}}\, a_{\mathbf{t}, h}^{(\gamma,e)}
\end{align}
and
\begin{align}\label{eq:a_VN}
a_{\mathbf{t}, h}^{(\gamma,e)} & = (q_{\gamma,e} - t_e) \,
\tilde{e}^{(\gamma)}_{q_{\gamma,1}-t_1, \dots, q_{\gamma,n_e}-t_{n_e}; b_{\gamma}-h}  \notag \\
\, & - (t_e+1) \, \tilde{e}^{(\gamma)}_{q_{\gamma,1}-t_1, \dots,
q_{\gamma,e}-t_{e}-1, \dots, q_{\gamma,n_e}-t_{n_e}; b_{\gamma}-h}.
\end{align}
Moreover, we have
$\mathsf{I}_{\mathsf{EC},e}(I_{AC,1},I_{AC,2},\dots,I_{AC,{n_e}})=\sum_{
\delta\in\scrF_C} \rho_{\delta,e} \,
\mathsf{I}_{\mathsf{EC},e}^{(\delta)}(I_{AC,1},I_{AC,2},\dots,I_{AC,{n_e}})$, where
\begin{align}\label{eq:GC102}
\mathsf{I}_{\mathsf{EC},e}^{(\delta)} & (I_{AC,1},I_{AC,2},\dots,I_{AC,{n_e}}) \notag \\ 
\, & = 1 -
\frac{1}{s_{\delta,e}} \sum_{t_1=0}^{s_{\delta,1}}
(1-I_{AC,1})^{t_1} (I_{AC,1})^{s_{\delta,1}-t_1}
\cdots \notag \\
\, & \times \sum_{t_e=0}^{s_{\delta,e}-1} (1-I_{AC,e})^{t_e}
(I_{AC,e})^{s_{\delta,e}-1-t_e} \cdots \notag \\
\, & \times  
\sum_{t_{n_e}=0}^{s_{\delta,n_e}} (1-I_{AC,n_e})^{t_{n_e}}
(I_{AC,n_e})^{s_{\delta,n_e}-t_{n_e}}\, a_{\mathbf{t}}^{(\delta,e)}
\end{align}
and
\begin{align}\label{eq:a_CN}
a_{\mathbf{t}}^{(\delta,e)} & = (s_{\delta,e} - t_e) \,
\tilde{e}^{(\delta)}_{s_{\delta,1}-t_1, \dots, s_{\delta,n_e}-t_{n_e}} \notag \\
\, & - (t_e+1) \, \tilde{e}^{(\gamma)}_{s_{\delta,1}-t_1, \dots,
s_{\delta,e}-t_{e}-1, \dots, s_{\delta,n_e}-t_{n_e}}.
\end{align}

\medskip
\begin{lemma}\label{lemma:fixed_point} For a MET D-GLDPC code ensemble such that all VNs and CNs
have local minimum distance
at least $2$ and such that no encoded bit is punctured ($\bfb_{\gamma}$ is the all-$1$
vector for all $\gamma\in\scrF_V$),
$[I_{EV,{1}},I_{EV,{2}},\dots,I_{EV,{n_e}}]^T=\bfone$ is
a fixed point of \eqref{eq:exit_compact} regardless of $\epsilon$, i.e., $\bff
(\bfone,\epsilon)=\bfone$ $\forall\epsilon\in(0,1)$.\footnote{The proof of
Lemma~\ref{lemma:fixed_point} is omitted due to space constraints. The lemma can be easily proved
by proving that, in \eqref{eq:GC101} and \eqref{eq:GC102}, under the two mentioned hypotheses,
$\mathsf{I}^{(\gamma)}_{\mathsf{EV},e} \rightarrow 1$ as $(I_{AV,1},I_{AV,2},\dots,I_{AV,n_e})
\rightarrow (1,1,\dots,1)$ $\forall (\gamma,e) \in \scrF_V \times \scrE$ and
$\mathsf{I}^{(\delta)}_{\mathsf{EC},e} \rightarrow 1$ as $(I_{AC,1},I_{AC,2},\dots,I_{AC,n_e})
\rightarrow (1,1,\dots,1)$ $\forall (\delta,e) \in \scrF_C \times \scrE$.}
\end{lemma}

\medskip
The fixed point $\mbox{\boldmath $I$}_{EV}=\bfone$ corresponds to a state of the system in
which no erasure messages are exchanged between the VN set and the CN set, i.e., in which all
encoded bits of the D-GLDPC code are known. A transmission over the BEC may be then modeled as a
perturbation of the system state from $\mbox{\boldmath $I$}_{EV}=\bfone$ to $\mbox{\boldmath
$I$}_{EV}=[\mathsf{I}_{\mathsf{EV},1}(0,0,\dots,0,\epsilon),\dots,\mathsf{I}_{\mathsf{EV},n_e}(0,0,
\dots,0,\epsilon)]^T:=\mbox{\boldmath $I$}^0_{EV}(\epsilon)$, and the corresponding evolution of the
BP decoder is modeled by \eqref{eq:exit_compact}. Decoding is successful when,
starting from $\mbox{\boldmath $I$}^0_{EV}(\epsilon)$, we have
$\lim_{\ell\rightarrow\infty}\mbox{\boldmath $I$}^{\ell}_{EV}=\bfone$. In order for the limit to be
$\bfone$, it is necessary that the steady-state equilibrium $\mbox{\boldmath $I$}_{EV}=\bfone$
acts as a \emph{local attractor} for the system, or, equivalently, that it is \emph{locally
stable}. The stability condition is established in the following theorem, which represents the main
contribution of this paper.

%

\medskip
\begin{theorem}\label{theorem:J=PC} Consider a MET D-GLDPC code ensemble with $n_e$ edge types.
Assume that there are no VNs and CNs with local minimum distance $1$ and that no encoded bit of the
D-GLDPC code is punctured ($\bfb_{\gamma}$ is the all-$1$ vector for all $\gamma\in\scrF_{V}$).
Define $\mbox{\boldmath $C$}$ as the $(n_e \times n_e)$ nonnegative matrix whose $(l,m)$-th entry is
$C^{l,m}$ in \eqref{eq:Clm}. Moreover, define $\mbox{\boldmath $P$}(\epsilon)$ as the $(n_e \times
n_e)$ nonnegative matrix of polynomials whose $(l,m)$-th entry is $P^{l,m}(\epsilon)$ in
\eqref{eq:Plm}. Then, the fixed point $\mbox{\boldmath $I$}_{EV}=\bfone$ of \eqref{eq:exit_compact}
is locally stable if and only if
\begin{align}\label{eq:stability_condition}
\sigma\left(\mbox{\boldmath $P$}(\epsilon) \, \mbox{\boldmath $C$}\right)<1
\end{align}
being $\sigma(\mbox{\boldmath $A$})$ the spectral radius of a square matrix
\mbox{\boldmath $A$}, i.e., the magnitude of the eigenvalue of \mbox{\boldmath $A$} with the largest
magnitude.
\end{theorem}

\medskip
Interestingly, inequality \eqref{eq:stability_condition} represents the ``natural'' extension to
the MET framework of the condition $P(\epsilon)C<1$ proved in \cite{paolini09:stability} for the
single-edge type case. A sketch of proof of Theorem~\ref{theorem:J=PC} is provided in
Section~\ref{section:sketch}.
Theorem~\ref{theorem:J=PC} allows us to develop a simple sufficient condition
for local stability of fixed point $\mbox{\boldmath $I$}_{EV}=\bfone$, as follows.

\medskip
\begin{corollary}
Consider a MET D-GLDPC code ensemble with $n_e$ edge types. Assume that there are no VNs and CNs
with local minimum distance $1$ and that no D-GLDPC encoded bit is punctured. Moreover, assume that
the follwing condition is satisfied: If a socket of VN of type $\gamma\in\scrF_{V2}$, associated
with the support of a weight-$2$ local codeword, is of type $l\in\scrE$, then for all
$\delta\in\scrF_{C2}$ a CN of type $\delta$ has no sockets of type $l$ associated with the support
of a weight-$2$ local codeword. Then, the fixed point $\mbox{\boldmath $I$}_{EV}=\bfone$ of
\eqref{eq:exit_compact} is locally stable for any BEC erasure probability $\epsilon$.
\end{corollary}
\begin{proof}
Simply observe that in this case $\mbox{\boldmath $P$}(\epsilon) \, \mbox{\boldmath $C$}$ is the
all-zero matrix and, consequently, $\sigma(\mbox{\boldmath $P$}(\epsilon) \, \mbox{\boldmath
$C$})=0$ for all $\epsilon$.
\end{proof}

\section{Sketch of Proof of Theorem \ref{theorem:J=PC}}\label{section:sketch}
For ease of presentation, we consider the case $n_e=2$. The extension of the proof to the case of
$n_e>2$ edge-types is straightforward. 

It is well-known that the local stability of a fixed point of a multidimensional discrete dynamical
system such as \eqref{eq:exit_compact} depends on the eigenvalues of the Jacobian matrix calculated
in the fixed point. Specifically, the fixed point is a local attractor when the magnitude
of all eigenvalues of the Jacobian matrix is less than $1$ or, equivalently, if and only if the
spectral radius of the Jacobian matrix is less than $1$. Hence, we need to
prove that $\mbox{\boldmath $J$}(\mbox{\boldmath $1$},\epsilon)=\mbox{\boldmath $P$}(\epsilon) \,
\mbox{\boldmath $C$}$, where $\mbox{\boldmath $J$}(\mbox{\boldmath $1$},\epsilon)$ is the $(n_e
\times n_e)$ Jacobian matrix of $\mbox{\boldmath $f$}(\mbox{\boldmath $I$}_{EV},\epsilon)$,
calculated in $\mbox{\boldmath $I$}_{EV}=\bfone$.

For $l,m\in\{1,2\}$,
the $(l,m)$-th entry of $\mbox{\boldmath $J$}(\mbox{\boldmath $1$},\epsilon)$, is given~by
\begin{align}\label{eq:Jlm_expanded}
J^{l,m}(\mbox{\boldmath $1$},\epsilon) = \sum_{e=1}^{2} \frac{\partial\,
\mathsf{I}_{\mathsf{EV},l}}{\partial I_{AV,e}}(\mbox{\boldmath $1$},\epsilon) \cdot \frac{\partial\,
\mathsf{I}_{\mathsf{EC},e}}{\partial I_{AC,m}}(\mbox{\boldmath $1$}) \, .
\end{align}

Consider now a generic $(q_{\gamma},k_{\gamma})$ VN of type $\gamma\in\scrF_V$ having at least one
codeword socket of edge type $e$. Using \eqref{eq:GC101} and \eqref{eq:a_VN}, it is easy to show
that
\begin{align}\label{eq:lim_IEVl_IAVe}
\, & \lim_{\scriptsize{\mbox{\boldmath $I$}_{AV} \rightarrow \mbox{\boldmath $1$}}} \frac{\partial\,
\mathsf{I}_{\mathsf{EV},l}^{(\gamma)}}{\partial I_{AV,e}^{(\gamma)}}(\mbox{\boldmath
$I$}_{AV},\epsilon) \notag \\ 
& = -\frac{1}{q_{\gamma,l}} \sum_{z=0}^{k_{\gamma}} \epsilon^z (1-\epsilon)^{k_{\gamma}-z}
[(q_{\gamma,e}-1) a^{(\gamma,l)}_{\mbox{\boldmath \scriptsize{$0$}},z} -
a^{(\gamma,l)}_{\mbox{\boldmath \scriptsize{$1$}}_e,z}] \notag \\
\, & = 
\frac{1}{q_{\gamma,l}} \sum_{z=0}^{k_{\gamma}} \epsilon^z (1-\epsilon)^{k_{\gamma}-z}
a^{(\gamma,l)}_{\mbox{\boldmath \scriptsize{$1$}}_e,z}
\end{align}
where the last equality is due to $a^{(\gamma,l)}_{\mbox{\boldmath \scriptsize{$0$}},z}=0$. In fact,
for $l=1$ we have $a^{(\gamma,1)}_{\mbox{\boldmath \scriptsize{$0$}},z}$ $= q_{\gamma,1}\cdot
\tilde{e}^{(\gamma)}_{q_{\gamma,1},q_{\gamma,2},k_{\gamma}-z}-\tilde{e}^{(\gamma)}_{q_{\gamma,1}-1,
q_ { \gamma , 2 } , k_ { \gamma }
-z}=q_{\gamma,1}{k_{\gamma,1} \choose z} k_{\gamma,1} - q_{\gamma,1}{k_{\gamma,1} \choose z}
k_{\gamma,1}=0$, and in an analogous way we can show that $a^{(\gamma,2)}_{\mbox{\boldmath
\scriptsize{$0$}},z}=0$. Next, we develop \eqref{eq:lim_IEVl_IAVe}, assuming $l=1$, in the two cases
$\bfonee=[1,0]^T$ and $\bfonee=[0,1]^T$. From \eqref{eq:a_VN} and from the definition of multi-type
information function in Section~\ref{subsec:multi_type_IF}, we have
\begin{align}\label{eq:a10u}
\, & \sum_{z=0}^{k_{\gamma}} \epsilon^z (1-\epsilon)^{k_{\gamma}-z} a^{(\gamma,1)}_{[1,0],z} =
\sum_{z=0}^{k_{\gamma}} \epsilon^z (1-\epsilon)^{k_{\gamma}-z} \notag \\ 
\, & \times \left[
(q_{\gamma,1}-1)\tilde{e}^{(\gamma)}_{q_{\gamma,1}-1,q_{\gamma,2},k_{\gamma}-z}-2\,
\tilde{e}^{(\gamma)}_{q_{\gamma,1}-2,q_{\gamma,2},k_{\gamma}-z} \right] \notag \\
\, & = 2 \sum_{z=0}^{k_{\gamma}} \epsilon^z (1-\epsilon)^{k_{\gamma}-z} \notag \\ 
\, & \times \Big[ \frac{q_{\gamma,1}(q_{\gamma,1}-1)k_{\gamma}{k_{\gamma}\choose
z}}{2}-\tilde{e}^{(\gamma)}_{q_{\gamma,1}-2,q_{\gamma,2},k_{\gamma}-z}\Big] \notag \\
\, & = 2 \sum_{z=0}^{k_{\gamma}} \epsilon^z (1-\epsilon)^{k_{\gamma}-z} \notag \\
\, & \sum_{\bfS_{q_{\gamma,1}-2,q_{\gamma,2};k_{\gamma}-z}} \left( k_{\gamma} - \mathrm{rank}
(\bfS_{q_{\gamma,1}-2,q_{\gamma,2};k_{\gamma}-z}) \right)
\end{align}
and
\begin{align}\label{eq:a01u}
\, & \sum_{z=0}^{k_{\gamma}} \epsilon^z (1-\epsilon)^{k_{\gamma}-z} a^{(\gamma,1)}_{[0,1],z} =
\sum_{z=0}^{k_{\gamma}} \epsilon^z (1-\epsilon)^{k_{\gamma}-z} \notag \\ 
\, & \times \left[
q_{\gamma,1}\tilde{e}_{q_{\gamma,1},q_{\gamma,2}-1,k_{\gamma}-z}-\tilde{e}_{q_{\gamma,1}-1,q_{
\gamma,2}-1,k_{\gamma}-z} \right] \notag \\
\, & = \sum_{z=0}^{k_{\gamma}} \epsilon^z (1-\epsilon)^{k_{\gamma}-z} \notag \\ 
\, & \times \Big[ q_{\gamma,1}\, q_{\gamma,2}\, k_{\gamma}{k_{\gamma}\choose
z}-\tilde{e}_{q_{\gamma,1}-1,q_{\gamma,2}-1,k_{\gamma}-z} \Big] \notag \\
\, & = \sum_{z=0}^{k_{\gamma}} \epsilon^z (1-\epsilon)^{k_{\gamma}-z} \notag \\
\, & \times \sum_{\bfS_{q_{\gamma,1}-1,q_{\gamma,2}-1;k_{\gamma}-z}} \left( k_{\gamma} -
\mathrm{rank}
(\bfS_{q_{\gamma,1}-1,q_{\gamma,2}-1;k_{\gamma}-z}) \right) \, .
\end{align}
It is possible to show that \eqref{eq:a10u} and \eqref{eq:a01u} are equivalent to
$2\sum_{u=1}^{k_{\gamma}} B_{2,u}^{(\gamma)}(1,1) \epsilon^u = \sum_{u=1}^{k_{\gamma}}
\chi_{2,u}^{(\gamma)}(1,1) \epsilon^u$ and
$\sum_{u=1}^{k_{\gamma}} B_{2,u}^{(\gamma)}(1,2)\epsilon^u = \sum_{u=1}^{k_{\gamma}}
\chi_{2,u}^{(\gamma)}(1,2) \epsilon^u$, respectively. Both expressions are obtained
through an argument along the same line as that used, in the one-edge type case, to prove
Lemma~4 in \cite{paolini09:stability}. Incorporating these expressions into
\eqref{eq:lim_IEVl_IAVe}, recalling that $\mathsf{I}_{\mathsf{EV},1}(\mbox{\boldmath
$I$}_{AV},\epsilon)=\sum_{\gamma\in\scrF_V}\lambda_{\gamma,1}
\mathsf{I}^{(\gamma)}_{\mathsf{EV},1}(\mbox{\boldmath
$I$}_{AV},\epsilon)$, and recalling \eqref{eq:Plm}, we finally obtain
\begin{align}\label{eq:P11}
\frac{\partial\, \mathsf{I}_{\mathsf{EV},1}}{\partial I_{AV,1}}(\mbox{\boldmath $1$},\epsilon) & =
\sum_{\gamma\in\scrF_{V2}} \left( \frac{\lambda_{\gamma,1}}{q_{\gamma_1}} \right)
\sum_{u=1}^{k_{\gamma}} \chi_{2,u}^{(\gamma)}(1,1) \epsilon^u = P^{1,1}(\epsilon)
\end{align}
%
%
\begin{align}\label{eq:P12}
\frac{\partial\, \mathsf{I}_{\mathsf{EV},1}}{\partial I_{AV,2}}(\mbox{\boldmath $1$},\epsilon) & =
\sum_{\gamma\in\scrF_{V2}} \left( \frac{\lambda_{\gamma,1}}{q_{\gamma_1}} \right)
\sum_{u=1}^{k_{\gamma}} \chi_{2,u}^{(\gamma)}(1,2) \epsilon^u P^{1,2}(\epsilon)\, .
\end{align}
Note that in both \eqref{eq:P11} and \eqref{eq:P12} the summation is over $\scrF_{V2}$ since, for
any $\gamma\in\scrF_{V}\setminus\scrF_{V2}$, we have
$\chi_{2,u}^{(\gamma)}(1,1)=\chi_{2,u}^{(\gamma)}(1,2)=0$. The same proof technique leading to
\eqref{eq:P11} and \eqref{eq:P12} yields $\frac{\partial\, \mathsf{I}_{EV,2}}{\partial
I_{AV,2}}(\mbox{\boldmath $1$},\epsilon)=P^{2,2}(\epsilon)$ and $\frac{\partial\,
\mathsf{I}_{EV,2}}{\partial I_{AV,1}}(\mbox{\boldmath $1$},\epsilon)=P^{2,1}(\epsilon)$.

We now need to develop $\frac{\partial\, \mathsf{I}_{\mathsf{EC},e}}{\partial
I_{AC,m}}(\mbox{\boldmath $1$})$ in the
right-hand side of \eqref{eq:Jlm_expanded}. To this purpose, simply observe that a CN of type
$\delta\in\scrF_{C}$ may be regarded as a VN whose $k_{\delta}$ local information bits are all
punctured ($\bfb_{\delta}=\bfzero$). Note that this is equivalent to assuming a channel erasure
probability $\epsilon=1$ for the VN. In this case, the right-hand sides of \eqref{eq:P11} and
\eqref{eq:P12} become $\sum_{\gamma\in\scrF_{V2}} \frac{\lambda_{\gamma,1}}{q_{\gamma_1}}
\sum_{u=1}^{k_{\gamma}}
\chi_{2,u}^{(\gamma)}(1,1) = \sum_{\gamma\in\scrF_{V2}} 
\frac{\lambda_{\gamma,1}}{q_{\gamma_1}} 
\chi_{2}^{(\gamma)}(1,1)$ and $\sum_{\gamma\in\scrF_{V2}} \!
\frac{\lambda_{\gamma,1}}{q_{\gamma_1}} \sum_{u=1}^{k_{\gamma}}
\chi_{2,u}^{(\gamma)}(1,2) = \sum_{\gamma\in\scrF_{V2}} 
\frac{\lambda_{\gamma,1}}{q_{\gamma_1}} 
\chi_{2}^{(\gamma)}(1,2)$, respectively. Thus, we have
\begin{align}\label{eq:C11}
\frac{\partial\, \mathsf{I}_{\mathsf{EC},1}}{\partial I_{AC,1}}(\mbox{\boldmath $1$}) & =
\sum_{\delta\in\scrF_{C2}} \left( \frac{\rho_{\delta,1}}{s_{\delta_1}}
\right) \xi_{2}^{(\delta)}(1,1) = C^{1,1} \, ,
\end{align}
%
%
\begin{align}\label{eq:C12}
\frac{\partial\, \mathsf{I}_{\mathsf{EC},1}}{\partial I_{AC,2}}(\mbox{\boldmath $1$}) & =
\sum_{\delta\in\scrF_{C2}} \left( \frac{\lambda_{\delta,1}}{q_{\delta_1}} \right)
\xi_{2}^{(\delta)}(1,2) = C^{1,2} \, ,
\end{align}
and also $\frac{\partial\, \mathsf{I}_{\mathsf{EC},2}}{\partial I_{AC,2}}(\mbox{\boldmath
$1$})=C^{2,2}$ and $\frac{\partial\, \mathsf{I}_{\mathsf{EC},2}}{\partial I_{AC,1}}(\mbox{\boldmath
$1$})=C^{2,1}$. Hence, for $l,m\in\{1,2\}$, the $(l,m)$-th entry of $\mbox{\boldmath
$J$}(\mbox{\boldmath $1$},\epsilon)$ is
given by $J^{l,m}(\mbox{\boldmath $1$},\epsilon) = \sum_{e=1}^2 P^{l,e}(\epsilon) C^{e,m}$, i.e.,
$\mbox{\boldmath $J$}(\mbox{\boldmath $1$},\epsilon)=\mbox{\boldmath $P$}(\epsilon) \,
\mbox{\boldmath $C$}$. \hfill\QED

\section{Examples}
In this section, the stability of the iterative decoder is analyzed for two MET ensembles. 

\begin{example}\label{example:product}
Consider the two-edge-type ensemble ($\mathscr{E}=\{1,2\}$) whose Tanner graph is depicted in
Fig.~\ref{fig:ex1}, where edges of type $1\in\mathscr{E}$ are depicted in red and
edges of type $2\in\mathscr{E}$ in blue. There are $N$ VNs, all of
the same type $\gamma$. Each VN is a length-$2$ repetition code with $\mathbf{G}_{\gamma}=[1,1]$,
with one socket of type $1\in\mathscr{E}$ and the other of type $2\in\mathscr{E}$. Thus, we have
$\mathscr{F}_V=\mathscr{F}_{V2}=\{\gamma\}$. There are two CN types
$\mathscr{F}_C=\{\delta_1,\delta_2\}$, where CNs of type $\delta_1$ are $(s_1,h_1)$ codes, depicted
in yellow, and CNs of type $\delta_2$ are $(s_2,h_2)$ codes, depicted in green. All $s_1$ sockets
of a type-$\delta_1$ CN are of type $1\in\mathscr{E}$, while all $s_2$ sockets of a type-$\delta_2$
CN are of type $2\in\mathscr{E}$. The number of CNs of types $\delta_1$ and $\delta_2$ are $N/s_1$
and $N/s_2$ respectively, so each edge interleaver is for $N$ edges. 

Assuming that CNs of both types have minimum distance $2$
($\mathscr{F}_{C2}=\{\delta_1,\delta_2\}$), we obtain
$$
\mbox{\boldmath $P$}(\epsilon)=\left[ \begin{array}{cc} 0 & \epsilon\\ \epsilon & 0 \end{array}
\right] \quad
\mathrm{and} \quad \mbox{\boldmath $C$} = \left[ \begin{array}{cc} 2\,A_2^{(\delta_1)}/s_1 & 0\\ 0 &
2\,A_2^{(\delta_2)}/s_2 \end{array} \right] \, ,
$$
where $A_2^{(\delta_1)}$ and $A_2^{(\delta_2)}$ are the multiplicities of weight-$2$ local
codewords of CNs of types $\delta_1$ and $\delta_2$, respectively. From Theorem~\ref{theorem:J=PC},
the condition for local stability of the erasure-free state is
\begin{align}\label{stab:product}
\epsilon < \frac{1}{2} \sqrt{\frac{s_1 s_2}{A_2^{(\delta_1)} A_2^{(\delta_2)}}} \, ,
\end{align}
where the right-hand side is an upper bound on the iterative decoding threshold called the
\emph{stability bound}.

From \eqref{stab:product}, we see how the multiplicities $A_2^{(\delta_1)}$ and $A_2^{(\delta_2)}$
may jeopardize the decoder stability, and how increasing $s_1$ or $s_2$ is beneficial in terms of
stability. We may also observe that the erasure-free fixed point for this ensemble is a stable
attractor if the CNs of at least one type are characterized by minimum distance larger than $2$,
\emph{irrespective} of the local weight spectrum of CNs of the other type (all diagonal
entries as well as at least one off-diagonal entry of $C$ are zero in this case). In practice,
for large $N$ this model gives a good indication of stability for the ensemble of \emph{product
codes} which are obtained by taking $N=s_1 s_2$ and choosing appropriately the two edge
interleavers $\Pi_1$ and $\Pi_2$ \cite{lentmaier10:product}.
\end{example}

\begin{figure}[t]
\begin{center}
\includegraphics[width=0.63\columnwidth]{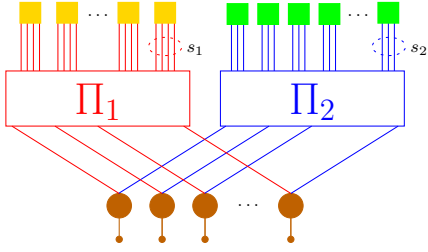}
\end{center}
\caption{MET ensemble analyzed in Example~\ref{example:product}.}
\label{fig:ex1}
\end{figure}

\begin{example}\label{example:gen_RA}
Consider the two-edge-type ensemble ($\mathscr{E}=\{1,2\}$) whose Tanner graph is depicted in
Fig.~\ref{fig:ex2}, where edges of type $1\in\mathscr{E}$ are depicted in red and edges of type
$2\in\mathscr{E}$ in blue. There are two VN types $\mathscr{F}_V=\{\gamma_1,\gamma_2\}$, where the
$N_{\gamma_1}$ VNs of type $\gamma_1$, depicted in cyan, are $(q,k)$ codes generated by some
generator matrix $\mathbf{G}_{\gamma_1}$, and the $N_{\gamma_2}=N_{\gamma_1}\, q$ VNs of type
$\gamma_2$, depicted in pink, are length-$2$ repetition codes with $\mathbf{G}_{\gamma_2}=[1,1]$.
All $q$ sockets of a type-$\gamma_1$ VN are of type $1\in\mathscr{E}$, while both sockets of a
type-$\gamma_2$ VN are of type $2\in\mathscr{E}$. Moreover, there are $N_{\gamma_1}\, q$ CNs, all
of the same type $\delta$. Each CN is a $(3,2)$ SPC code having one socket of type
$1\in\mathscr{E}$ and two sockets of type $2\in\mathscr{E}$. Hence, we have
$\mathscr{F}_C=\mathscr{F}_{C2}=\{\delta\}$.

Assuming that VNs of type $\gamma_1$ have minimum distance $2$
($\mathscr{F}_{V2}=\{\gamma_1,\gamma_2\}$), we obtain
$$
\mbox{\boldmath $P$}(\epsilon)=\left[ \begin{array}{cc} \frac{2}{q}\sum_{u=1}^k
B^{(\gamma_1)}_{2,u}\,
\epsilon^u & 0 \\ 0 & \epsilon \end{array}
\right] \quad
\mathrm{and} \quad \mbox{\boldmath $C$} = \left[ \begin{array}{cc} 0 & 2\\ 1 & 1 \end{array} \right]
\, ,
$$
where $B^{(\gamma_1)}_{2,u}$ denotes the number of weight-$2$ local codewords of VNs of type
$\gamma_1$
generated by local input words of length $k$ through $\mathbf{G}_{\gamma_1}$. Again applying
Theorem~\ref{theorem:J=PC}, we obtain the following condition for stability of the erasure-free
fixed point:
\begin{align}\label{eq:stab_gen_RA}
\frac{4}{q} \sum_{u=1}^k B^{(\gamma_1)}_{2,u}\, \epsilon^{u+1} < 1-\epsilon \, .
\end{align}

Again, an increase in the multiplicity of weight-$2$ local codewords of VNs of type $\gamma_1$ has
a negative effect on the stability of the fixed point $\mbox{\boldmath $I$}_{EV}=\bfone$, as it
reduces the range of channel erasure probabilities over which such a fixed point is locally stable
(just note that all coefficients of the polynomial on the left-hand side of
\eqref{eq:stab_gen_RA} are positive). Moreover, increasing $q$ has a positive effect on the
stability. We also point out that the fixed point $\mbox{\boldmath $I$}_{EV}=\bfone$ \emph{must} be
locally stable if the minimum distance of the type-$\gamma_1$ VNs is larger than $2$
(in fact, in this case we obtain $\epsilon < 1$). Finally, we observe that, upon a proper choice of
the edge interleaver $\Pi_2$, the Tanner graph depicted in Fig.~\ref{fig:ex2} corresponds to the
serial concatenation of a $(q,k)$ linear block encoder $\mathbf{G}_{\gamma_1}$ with an
accumulator. Hence, this class of codes may be seen as a generalization of repeat-accumulate (RA)
codes \cite{divsalar98:RA}. An RA code is obtained when type-$\gamma_1$ VNs are length-$q$
repetition codes.
\end{example}

\begin{figure}[t]
\begin{center}
\includegraphics[width=0.65\columnwidth]{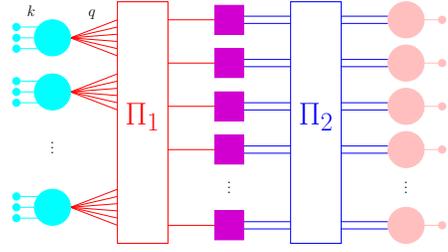}
\end{center}
\caption{MET ensemble analyzed in Example~\ref{example:gen_RA}.}
\label{fig:ex2}
\end{figure}

\vspace{4mm}
\section{Conclusion}
In this paper, the stability condition for iterative BP decoding of MET D-GLDPC codes over the BEC
has been developed. The obtained inequality is compact, and naturally extends to the MET ensemble
parametrization the previously obtained condition for unstructured irregular (single-edge type)
D-GLDPC codes. Although this point is not addressed in the present
work, we mention that for LDPC-like codes, the stability condition has
a further practical impact on code design through its relationship
with the average weight distribution of the ensemble. 

\bibliographystyle{IEEEtran}
\vspace{3.5mm}
\bibliography{IEEEabrv,bibfile}

\begin{thebibliography}{10}
\providecommand{\url}[1]{#1}
\csname url@samestyle\endcsname
\providecommand{\newblock}{\relax}
\providecommand{\bibinfo}[2]{#2}
\providecommand{\BIBentrySTDinterwordspacing}{\spaceskip=0pt\relax}
\providecommand{\BIBentryALTinterwordstretchfactor}{4}
\providecommand{\BIBentryALTinterwordspacing}{\spaceskip=\fontdimen2\font plus
\BIBentryALTinterwordstretchfactor\fontdimen3\font minus
  \fontdimen4\font\relax}
\providecommand{\BIBforeignlanguage}[2]{{%
\expandafter\ifx\csname l@#1\endcsname\relax
\typeout{** WARNING: IEEEtran.bst: No hyphenation pattern has been}%
\typeout{** loaded for the language `#1'. Using the pattern for}%
\typeout{** the default language instead.}%
\else
\language=\csname l@#1\endcsname
\fi
#2}}
\providecommand{\BIBdecl}{\relax}
\BIBdecl

\bibitem{richardson2004:MET}
T.~Richardson and R.~Urbanke, ``Multi-edge type {LDPC} codes,''
  LTHC-Report-2004-001, 2004.

\bibitem{kasai10:weight_MET}
K.~Kasai, T.~Awano, D.~Declercq, C.~Poulliat, and K.~Sakaniwa, ``Weight
  distributions of multi-edge type {LDPC} codes,'' \emph{IEICE Trans.
  Fundamentals}, vol. E93-A, no.~11, pp. 1942--1948, Nov. 2010.

\bibitem{luby01:improved}
M.~Luby, M.~Mitzenmacher, M.~Shokrollahi, and D.~Spielman, ``Improved
  low-density parity-check codes using irregular graphs,'' \emph{{IEEE} Trans.
  Inf. Theory}, vol.~47, no.~2, pp. 585--598, Feb. 2001.

\bibitem{thorpe2003:protograph}
J.~Thorpe, ``Low density parity check ({LDPC}) codes constructed from
  protographs,'' JPL INP Progress Report 42-154, 2003.

\bibitem{wang06:D-GLDPC}
Y.~Wang and M.~Fossorier, ``Doubly generalized low-density parity-check
  codes,'' in \emph{Proc. of 2006 IEEE Int. Symp. Inf. Theory}, Seattle, WA,
  USA, Jul. 2006, pp. 669--673.

\bibitem{tanner81:recursive}
R.~M. Tanner, ``A recursive approach to low complexity codes,'' \emph{{IEEE}
  Trans. Inf. Theory}, vol.~27, no.~5, pp. 533--547, Sep. 1981.

\bibitem{paolini09:stability}
E.~Paolini, M.~Fossorier, and M.~Chiani, ``Doubly-generalized {LDPC} codes:
  {S}tability bound over the {BEC},'' \emph{{IEEE} Trans. Inf. Theory},
  vol.~55, no.~3, pp. 1027--1046, Mar. 2009.

\bibitem{flanagan11:growth_rate}
M.~Flanagan, E.~Paolini, M.~Chiani, and M.~Fossorier, ``On the growth rate of
  the weight distribution of irregular doubly-generalized {LDPC} codes,''
  \emph{{IEEE} Trans. Inf. Theory}, vol.~57, no.~6, pp. 3721--3737, Jun. 2011.

\bibitem{fossorier09:MET}
C.-L. Wang, S.~Lin, and M.~Fossorier, ``On asymptotic ensemble weight
  enumerators of multi-edge type codes,'' in \emph{Proc. of 2009 IEEE Global
  Telecommun. Conf.}, Honolulu, HI, USA, Nov./Dec. 2009.

\bibitem{paolini10:multi}
E.~Paolini, M.~Chiani, and M.~Fossorier, ``On design of doubly-generalized
  {LDPC} codes based on multi-type information functions,'' in \emph{Proc. of
  2010 IEEE Global Telecommun. Conf.}, Miami, FL, USA, Dec. 2010.

\bibitem{urbanke:modern}
T.~Richardson and R.~Urbanke, \emph{Modern Coding Theory}.\hskip 1em plus 0.5em
  minus 0.4em\relax Cambridge University Press, 2008.

\bibitem{lentmaier10:product}
M.~Lentmaier, G.~Liva, E.~Paolini, and G.~Fettweis, ``From product codes to
  structured generalized {LDPC} codes,'' in \emph{{Proc. of the 5th Int. ICST
  Conf. Commun. Networking in China}}, Beijing, China, Aug. 2010.

\bibitem{divsalar98:RA}
D.~Divsalar, H.~Jin, and R.~J. McEliece, ``Coding theorems for turbo-like
  codes,'' in \emph{Proc. of 1998 Allerton Conf. Commun., Control Comput.},
  Monticello, IL, USA, Sep. 1998, pp. 201--210.

\end{thebibliography}
\end{document}